\documentclass{article}
\usepackage{amssymb}

\usepackage{amsmath}

\newtheorem{theorem}{Theorem}[section]

\newtheorem{conjecture}[theorem]{Conjecture}

\newtheorem{definition}[theorem]{Definition}

\newtheorem{lemma}[theorem]{Lemma}

\newtheorem{proposition}[theorem]{Proposition}
\newtheorem{remark}[theorem]{Remark}

\newtheorem{Main Result:}{Main Result:}
\newenvironment{proof}[1][Proof]{\textbf{#1.} }{\ \rule{0.5em}{0.5em}} %
\begin{document}

\title{{Note\\ On the conjecture about the nonexistence of rotation symmetric bent functions}
\thanks{This work was supported by NSF of China with contract No.
60803154}
\author{Zhang Xiyong\footnote{\textbf{Corresponding E-mail Address:} xiyong.zhang@hotmail.com}, \ \
Gao Guangpu \\
{\small Zhengzhou Information Science and Technology Institute, PO Box 1001-745, Zhengzhou 450002, PRC} \\
 }
\date{}}
\maketitle

\vspace{-0.4cm}

\begin{abstract}
In this paper, we describe a different approach to the proof of
the nonexistence of homogeneous rotation symmetric bent functions.
As a result, we obtain some new results which support the
conjecture made in this journal, i.e., there are no homogeneous
rotation symmetric bent functions of degree $>2$. Also we
characterize homogeneous degree $2$ rotation symmetric bent
functions by using GCD of polynomials.
\end{abstract}

\par \textbf{Keywords:}
{\textit{Boolean functions, Bent, Rotation-symmetric, Fourier
Transform}

\section{Motivation}
Since the introduction in the seventies by Rothaus \cite{rothus},
bent functions have been intensively studied in the past three
decades, and widely used in cryptography and error-correction
coding due to their nice cryptographic and combinatoric
properties. For example, the highest possible nonlinearity of bent
functions can be used to resist the differential attack and the
linear attack in symmetric cipher.

Recently, homogeneous rotation symmetric (Abbr. RotS) Boolean
functions have attracted attentions (see
\cite{pieprzyk,cusik,kavut}) because of their highly desirable
property, i.e.  they can be evaluated efficiently by re-using
evaluations from previous iterations. Consequently, when efficient
evaluation of the function (for example, design of some
cryptographic algorithm, such as MD4 and MD5) is essential, these
functions can serve as a good option.

It is natural to ask what kind of homogeneous RotS bent functions
exist. In fact, homogeneous bent functions are of interest in
literature \cite{charnes,stanica2002,stanica,xia,meng,dalai}.
St\u{a}nic\u{a} and Maitra \cite{stanica2002,stanica} studied RotS
bent functions up to 10-variables. They enumerated all RotS bent
functions in 8-variables. $4\cdot 3776$ such functions of degree 2
were found. However, they couldn't find any homogeneous RotS bent
functions of degree 3,4 and 5 in 10 variables. Thus they made the
following conjecture.

\begin{conjecture}
There are no homogeneous rotation symmetric bent functions of
degree $>2$.
\end{conjecture}

Let us summarize known results related to the above conjecture.
Observing that bent functions are in fact Hadamard difference
sets, Xia et al.\cite{xia} showed that there are no homogeneous
bent functions of degree $n$ in $2n$ variables for every $n>3$. By
using the relationship between the Fourier spectra of a Boolean
function at partial points and the Fourier spectra of its
sub-functions, Meng et al.\cite{meng} got a low bound of degree
for homogeneous bent functions. From the view point of
nonlinearity, St\u{a}nic\u{a} \cite{stanica2008} obtained the following
nonexistence results (see Section 2 for the notation SANF of a
Boolean function):

\begin{theorem}\label{dalai-stanica}
The following hold for a homogeneous RotS $f$ of degree $d\geq 3$
in $n$ variables:

(i) If the SANF of $f$ is $x_1\cdots x_d$, then $f$ is not bent.

(ii) If the SANF of $f$ is $x_1\cdots x_d+x_1\cdots x_{d-1}x_{d+1}$,
then $f$ is not bent, assuming: $\frac{n-2}{4}>\lfloor \frac{n}{d}
\rfloor$, if $n\neq 1 (mod\ d)$; $\frac{n}{4}>\lfloor \frac{n}{d}
\rfloor$, if $n\equiv 1 (mod\ d)$.

(iii) If the SANF of $f$ is
$\mathbf{x}^{\mathbf{u}_1}+\cdots+\mathbf{x}^{\mathbf{u}_m}$, then
$f$ is not bent when $d_f<\frac{n/2-1}{\lfloor n/d \rfloor}$,
where $d_f=Max_{i,j}\{j_2-j_1 | u_{ij_1}=u_{ij_2}=1,\and\
u_{ij}=0\ if\ j_1<j<j_2 \}$.
\end{theorem}

In this paper we will introduce another method which may be more
suitable for investigating homogeneous RotS
bent functions. By using the rotation symmetric forms of RotS
functions, we obtain more nonexistence results which are unaccessible by Theorem \ref{dalai-stanica}. For example our results
imply that most homogeneous degree $d(\geq 3)$ RotS bent functions
with SANF forms containing $x_1\cdots x_d$ cannot exist. Also, we
give an equivalent characterization of homogeneous degree 2 RotS
bent functions by GCD of polynomials.

\section{Preliminaries}

In this section  we list some basic definitions and notations about
homogeneous rotation symmetric Boolean functions and bent
functions.

Let $\mathbb{F}_2^n$ be the vector space of dimension $n$ over the
two element field $\mathbb{F}_2$. A Boolean function
$f(x_0,\cdots,x_{n-1})$ in $n$ variables is a map from
$\mathbb{F}_2^n$ to $\mathbb{F}_2$. For $\mathbf{x}=
(x_1,\cdots,x_n), \mathbf{u}= (u_1,\cdots,u_n)\in \mathbb{F}_2^n$.
Denote $\mathbf{x}^{\mathbf{u}}=x_1^{u_1}\cdots x_n^{u_n}$. Then
every Boolean function $f$ is uniquely of the form
$f(\mathbf{x})=\sum\limits_{\mathbf{u}\in\mathbf{U}_f}\mathbf{x}^\mathbf{u}$,
where $\mathbf{U}_f\subseteq \mathbb{F}_2^n$. Let $|\mathbf{u}|$
be the Hamming weight of $\mathbf{u}\in\mathbb{F}_2^n$. The
algebraic degree of $f$ is defined to be $Max\{|\mathbf{u}|\ |\
\mathbf{u}\in\mathbf{U}_f\}$.

By $A\|B$ we mean the concatenation of two bit strings $A$ and $B$. We
use $\underbrace{1\cdots 1}\limits_{l}$(respectively
$\underbrace{0\cdots 0}\limits_{l}$) to represent 1(respectively
0) string of length $l$, and $\underbrace{1*\cdots *1}\limits_{l}$
to represent a bit string of length $l$, with the first and the
last bit to be 1.

We define an operation $\oplus$ over $\mathbb{F}_2$ to be $x\oplus
y\in \mathbb{F}_2$ such that $x\oplus y=0$ if and only if $x=0$
and $y=0$. $\oplus$ can be extended to $\mathbb{F}_2^n$ by this
way: for $\mathbf{x},\mathbf{y}\in \mathbb{F}_2^n$,
$\mathbf{x}\oplus \mathbf{y}=(\mathbf{x}_1\oplus
\mathbf{y}_1,\cdots,\mathbf{x}_n\oplus \mathbf{y}_n)$. Let $1\leq
l\leq n$, the operation $\rho^l(\cdot)$ acing on $\mathbb{F}_2^n$
is defined to be
$\rho^l(x_1,\cdots,x_n)=(x_{n-l+1},x_{n-l+2},\cdots,x_n,x_1,\cdots,x_{n-l})$,
where $n+l=l$ if $l>0$. The cycle length $l_{\mathbf{x}}$ of
$\mathbf{x}\in \mathbb{F}_2^n$ is the least number $l$ such that
$\rho^l(\mathbf{x})=\mathbf{x}$. Obviously $l_{\mathbf{x}}|n$ and
$l_{\mathbf{x}}=l_{\rho(\mathbf{x})}$.

\begin{definition}\label{rots}
A Boolean function $f(\mathbf{x})$, is called rotation symmetric
(Abbr. RotS) if
$$f(\mathbf{x})=f(\rho(\mathbf{x})),\ for\ all\ \mathbf{x}\in \mathbb{F}_2^n.$$
\end{definition}

It is clear that a RotS function $f$ is of the form
$$f(\mathbf{x})=\sum\limits_{1\leq i\leq m}
\sum\limits_{0\leq l\leq
l_{\mathbf{u}_i}-1}\mathbf{x}^{\rho^l(\mathbf{u}_i)},$$ where
$m\geq 1, \mathbf{u}_i\in \mathbb{F}_2^n(1\leq i\leq m)$. Since
the existence of $\mathbf{x}^{\mathbf{u}_i}$ implies the existence
of $\mathbf{x}^{\rho(\mathbf{u}_i)}$, we can represent a RotS
function $f$ by the so-called \textit{short algebraic normal form}
(Abbr. SANF)
$\mathbf{x}^{\mathbf{u}_1}+\cdots+\mathbf{x}^{\mathbf{u}_m}$.

\begin{definition}\label{fourier}
For a Boolean function $f(x)$, the Fourier transform of $f$ at
$\mathbf{c}\in  \mathbb{F}_2^n$ is defined as
$$\widehat{f}(\mathbf{c})=\sum\limits_{\mathbf{x}\in \mathbb{F}_2^n}(-1)^{f(\mathbf{x})+\mathbf{c}\cdot \mathbf{x}},$$
where $\cdot$ is  dot product of two vectors in
$\mathbb{F}_2^n$.
\end{definition}

\begin{definition}\label{bent}
A Boolean function $f(\mathbf{x})$ is called bent if
$$|\widehat{f}(\mathbf{c})|=2^{n/2}\ \ for\ \ all\ \ \mathbf{c}\in \mathbb{F}_2^n.$$

\end{definition}

It is well-known that if $f(x_1,\cdots,x_n)$ is bent, then $n$
must be even, and the algebraic degree of $f$ is upper-bounded by
$n/2$.

Let
$f(\mathbf{x})=\mathbf{x}^{\mathbf{u}_1}+\cdots+\mathbf{x}^{\mathbf{u}_m}$,
define
$$h_f(\mathbf{u})=\sum\limits_{\stackrel{0\leq t_1,\cdots,t_m\leq 1}
{t_1\mathbf{u}_1\oplus\cdots \oplus
t_m\mathbf{u}_m=\mathbf{u}}}(-2)^{t_1+\cdots+t_m}.$$

It is not difficult to deduce that
$$\widehat{f}(\mathbf{c})=(-1)^{|\mathbf{c}|}\cdot
\sum\limits_{\mathbf{u}\succ
\mathbf{c}}2^{n-|\mathbf{u}|}h_f(\mathbf{u}),$$ where $\succ$ is a
partial order on $\mathbb{F}_2^n$ such that $(u_1,\cdots,u_n)\succ
(v_1,\cdots,v_n)$ if $u_i=v_i$ or $(u_i,v_i)=(1,0)$. We also have
the inverse formula:
$$h_f(\mathbf{u})=(-1)^{|\mathbf{u}|}\cdot 2^{|\mathbf{u}|-n}\cdot \sum\limits_{\mathbf{c}\succ \mathbf{u}}
\widehat{f}(\mathbf{c}).$$

By the above formulas and Definition \ref{bent}, one can prove
that
\begin{lemma}\cite{hou,carlet} \label{bentpanding}
Let $n$ be even and $f(\mathbf{x})=\mathbf{x}^{\mathbf{u}_1}
+\cdots+\mathbf{x}^{\mathbf{u}_m}$. Then $f$ is bent if and only
if
\begin{equation*}
v_2\left(h_f(\mathbf{u})\right)\left\{
\begin{array}{ll}
=n/2 &if\ \mathbf{u}=\mathbf{1},\\
>|\mathbf{u}|-n/2 &if\ \mathbf{u}
\neq\mathbf{1}.
\end{array}
\right.
\end{equation*}
where $v_2(\cdot)$ is the $2-$adic order function, $\mathbf{1}$
represents the vector $(1,\cdots,1)\in \mathbb{F}_2^n$.
\end{lemma}

\section{The result}
Let $f$ be a RotS function of homogeneous degree $d$, SANF of $f$
is $\sum\limits_{1\leq i\leq m}\mathbf{x}^{\mathbf{u}_i}$, where
$\mathbf{u}_i=(u_{i1},u_{i2},\cdots,u_{in}),u_{i1}=1$, and
$|\mathbf{u}_i|=d$.

We assume
\begin{equation*}
\begin{array}{ll}
D_i&=Max\{j|u_{ij}=1, 1\leq j\leq n\}, 1\leq i\leq m,\\
D_1&=Min\{D_i|1\leq i\leq m\}.\\
\end{array}
\end{equation*}

\begin{theorem}\label{multicycle}
Let $f$ be a RotS bent function of homogeneous degree $d\geq 3$,
the SANF of $f$ is $\sum\limits_{1\leq i\leq
m}\mathbf{x}^{\mathbf{u}_i}$, and
$\mathbf{u}_1=A_{l}\|B_{D_1-l}\|\underbrace{0\cdots0}\limits_{n-D_1},
1\leq l\leq D_1$, where $A_{l}=\underbrace{1*\cdots*1}\limits_{l},
B_{D_1-l}=\underbrace{1*\cdots*1}\limits_{D_1-l}$. If for all
$1\leq i\leq m$, $\mathbf{u}_i\neq
A_{l}\|\underbrace{0\cdots0}\limits_{D_i-D_1}\|B_{D_1-l}\|\underbrace{0\cdots0}\limits_{n-D_i}$
and $\mathbf{u}_i\neq
B_{D_1-l}\|\underbrace{0\cdots0}\limits_{n-kD_1}\|A_{l}\|\underbrace{0\cdots0}\limits_{kD_1-D_1}$,
where $kd<n$, then
$$k\cdot (d-1)< \frac{n}{2}.$$
\end{theorem}
\begin{proof}
Let
$$\mathbf{u}_0=\mathbf{u}_1\oplus\rho^{D_1}(\mathbf{u}_1)\oplus\cdots\oplus\rho^{(k-1)D_1}(\mathbf{u}_1),$$
where $kd<n$.

Because $|\mathbf{u}_i|=d$ for all $1\leq i\leq m$, we deduce that

$$Min\left\{\sum\limits_{1\leq i\leq m,1\leq j\leq n}e_{ij}\ |\ \bigoplus\limits_{1\leq i\leq m }\bigoplus\limits_{1\leq j\leq
n-1}e_{ij}\rho^{j}(\mathbf{u}_i) =\mathbf{u}_0, e_{ij}=0,1
\right\}=k.$$

Since for all $1\leq i\leq m$, $\mathbf{u}_i\neq
A_{l}\|\underbrace{0\cdots0}\limits_{D_i-D_1}\|B_{D_1-l}\|\underbrace{0\cdots0}\limits_{n-D_i}$
and $\mathbf{u}_i\neq
B_{D_1-l}\|\underbrace{0\cdots0}\limits_{n-kD_1}\|A_{l}\|\underbrace{0\cdots0}\limits_{kD_1-D_1}$,
the only solution such that $Min\left\{\sum\limits_{1\leq i\leq
m,1\leq j\leq n}e_{ij} \right\}=k$ for the  equation
$$ \bigoplus\limits_{1\leq i\leq m }\bigoplus\limits_{1\leq j\leq
n-1}e_{ij}\rho^{j}(\mathbf{u}_i) =\mathbf{u}_0, e_{ij}=0,1, $$ is
$$\mathbf{u}_0=\mathbf{u}_1\oplus\rho^{D_1}(\mathbf{u}_1)\oplus\cdots\oplus\rho^{(k-1)D_1}(\mathbf{u}_1).$$

Hence we get $v_2(h_f(\mathbf{u}_0))=k$. Since $kd<n$,
$\mathbf{u}_0\neq \mathbf{1}$. By Theorem \ref{bentpanding}, we
have
$$v_2(h_f(\mathbf{u}_0))=k>|\mathbf{u}|-\frac{n}{2}=kd-\frac{n}{2},$$
and thus $k\cdot (d-1)<\frac{n}{2}$.
\end{proof}

The above theorem implies nonexistence of many RotS bent
functions. For example, we get

\begin{proposition} For a RotS function $f$ of homogeneous degree
$d\geq 3$, the following nonexistence results hold,

(1) If the SANF $\sum\limits_{1\leq i\leq
m}\mathbf{x}^{\mathbf{u}_i}$ of $f$   contains
$\mathbf{x}^{\mathbf{u}_1}=x_1\cdots x_d$, and $\mathbf{u}_i(2\leq
i\leq m)$ is not of the form
$\underbrace{1\cdots1}\limits_{l}\underbrace{0\cdots0}\limits_{D_i-d}\underbrace{1\cdots1}
\limits_{d-l}\underbrace{0\cdots0}\limits_{n-D_i}$,  then $f$ is
not bent.

(2) If the SANF of $f$ is $x_1\cdots x_d+x_1\cdots
x_{d-1}x_{d+1}$, then $f$ is not bent.

(3) Suppose the SANF of $f$ be $x_1x_{2+n_1}x_{3+n_1+n_2}$ and $n=
q(D+n_0)+r+(n_1+1)$, where $n_1,n_2\geq 0,n_0=Max\{n_1,n_2\},
D=n_1+n_2+3, q\geq 1, 0\leq r<D+n_0$. If $q(D-n_0-1)\geq r+n_1+1$,
then $f$ is not bent.
\end{proposition}

\begin{proof}

(1)  If $d\nmid n$. Let $n=qd+r, 0< r< d, q\geq 2$.

Since $\mathbf{x}^{\mathbf{u}_1}=x_1\cdots x_d$, and
$\mathbf{u}_i(2\leq i\leq m)$ is not of the form
$\underbrace{1\cdots1}\limits_{l}\underbrace{0\cdots0}\limits_{D_i-d}\underbrace{1\cdots1}
\limits_{d-l}\underbrace{0\cdots0}\limits_{n-D_i}$ , using Theorem
\ref{multicycle}, we have
$$k\cdot (d-1)< \frac{n}{2},\ \ 1\leq k\leq \lfloor \frac{n}{d}\rfloor\ .$$

Let $k=\lfloor \frac{n}{d}\rfloor=q$. Thus $qd<2q+r<2q+d$, which
produces $d<2+\frac{2}{q-1}$.

\begin{enumerate}
    \item[(1.1)] If $q=2$, then $d<2+\frac{2}{q-1}=4$, so the only choice for $d$
is $d=3$. By again $qd<2q+r$, we have $r>3$, conflicting with
$r<d=3$.
    \item[(1.2)] If $q=3$, then $d<2+\frac{2}{q-1}=3$,  conflicting with $d\geq 3$.
    \item[(1.3)] If $q\geq 4$, then $d<2+\frac{2}{q-1}< 3$, conflicting with
$d\geq 3$.
\end{enumerate}

If $d|n$, let $n=qd,q\geq 2$. We choose $k=\lfloor
\frac{n}{d}\rfloor-1=q-1$. Similarly using Theorem
\ref{multicycle}, we have
$$(q-1)\cdot (d-1)< \frac{n}{2}=\frac{qd}{2},$$
which produces $d<2+\frac{2}{q-2}$.

\begin{enumerate}
    \item[(1.4)] If $q=2$, then $n=2d$. We choose another
$$\mathbf{u}_0=\mathbf{u}_1\oplus
\rho^{d-1}(\mathbf{u}_1)=\underbrace{1\cdots1}\limits_{n-1}\|0.$$
Similarly we get $v_2(h_f(\mathbf{u}_0))=2$. By Theorem
\ref{bentpanding},
$v_2(h_f(\mathbf{u}_0))=2>|\mathbf{u}_0|-n/2=n-1-n/2$. Therefore
$n<6$. So $d=n/2<3$, conflicting with $d\geq 3$.
    \item[(1.5)] If $q=3$, then $d<2+\frac{2}{q-2}=4$. Thus $d=3$ and $n=qd=9$.
Obviously this is impossible since a bent function of $n$
variables can exist for even $n$.
    \item[(1.6)] If $q\geq 4$, then $d<2+\frac{2}{q-2}< 3$, conflicting with
$d\geq 3$.
\end{enumerate}

(2) Denote $\mathbf{u}_1=\underbrace{1\cdots1}\limits_{d}\|
\underbrace{0\cdots0}\limits_{n-d}$,
$\mathbf{u}_2=\underbrace{1\cdots1}\limits_{d-1}\|0\|1\|
\underbrace{0\cdots0}\limits_{n-d-1}$, then the SANF of $f$ is
$\mathbf{x}^{\mathbf{u}_1}+\mathbf{x}^{\mathbf{u}_2}$.

If $n\neq 0,1(mod\ d)$, let $n=qd+r, 1< r< d, q\geq 2$. We choose
$$\mathbf{u}_0=\mathbf{u}_1\oplus\rho^{D_1}(\mathbf{u}_1)\oplus\cdots\oplus\rho^{(q-1)D_1}(\mathbf{u}_1)
=\underbrace{1\cdots1}\limits_{qd}\underbrace{0\cdots0}\limits_{r},$$
and it is easy to see that $v_2(h_f(\mathbf{u}_0))=q$.

Using Theorem \ref{multicycle}, we obtain $q\cdot (d-1)<
\frac{n}{2}$, and thus $qd<2q+r<2q+d$. The remaining discussions
are the same as (1.1),(1.2) and (1.3).

If $n\equiv 0(mod\ d)$, let $n=qd, q\geq 2$. We choose
$$\mathbf{u}_0=\mathbf{u}_1\oplus\rho^{D_1}(\mathbf{u}_1)\oplus\cdots\oplus\rho^{(q-2)D_1}(\mathbf{u}_1)
=\underbrace{1\cdots1}\limits_{(q-1)d}\underbrace{0\cdots0}\limits_{d}.$$
Similarly by Theorem \ref{multicycle} we get $d<2+\frac{2}{q-2}$.
We discuss the inequality in three cases: $q=2$, $q=3$ and $q\geq
4$. The proof for the cases $q=3,4$ are the same as (1.5), (1.6)
respectively.

If $q=2$, then $n=2d$. We choose
$$\mathbf{u}_0=\mathbf{u}_1\oplus\rho^{d-2}(\mathbf{u}_1)
=\underbrace{1\cdots1}\limits_{2d-2}\|00.$$ It is easy to see that
$v_2(h_f(\mathbf{u}_0))=2$. By Theorem \ref{bentpanding}, we get
$v_2(h_f(\mathbf{u}_0))=2>2d-2-n/2$. Thus $d<4$. Since $d\geq 3$,
we have $d=3,n=6$. However, RotS function over $6$ variables with
the SANF form $x_1x_2x_3+x_1x_2x_4$ can be verified to be
non-bent.

The remaining case is $n\equiv 1(mod\ d)$. Assume $n=qd+1$. We
choose
$$\mathbf{u}_0=\mathbf{u}_1\oplus\rho^{D_1}(\mathbf{u}_1)\oplus\cdots\oplus\rho^{(q-2)D_1}(\mathbf{u}_1)
=\underbrace{1\cdots1}\limits_{(q-1)d}\underbrace{0\cdots0}\limits_{d+1}.$$
Similarly  we get $d<2+\frac{1}{q-2}$. If $q>2$, then $d<3$, a
contradiction to $d\geq 3$. If $q=2$, then $n=2d+1$. However, a
Boolean functions in odd number variables cannot be bent.

(3) Denote $\mathbf{u}_1=1\|\underbrace{0\cdots0}\limits_{n_1}\|
1\|\underbrace{0\cdots0}\limits_{n_2}\|1\|\underbrace{0\cdots0}\limits_{n-D}$,
then the SANF of $f$ is $\mathbf{x}^{\mathbf{u}_1}$.

Since $n= q(D+n_0)+r+(n_1+1)$ and $q\geq 1$, we see that
$n-(n_1+1)\geq D+n_0$. Let $\mathbf{u}_2=\bigoplus\limits_{0\leq
i\leq n_0-1}\rho^i(\mathbf{u}_1)$, i.e.
\begin{equation*}
\begin{array}{lll}
\mathbf{u}_2

&=&1\underbrace{0\cdots0}\limits_{n_1}
1\underbrace{0\cdots0}\limits_{n_2}1\underbrace{0\cdots0}\limits_{n-D}\oplus\\
&\ &01\underbrace{0\cdots0}\limits_{n_1}
1\underbrace{0\cdots0}\limits_{n_2}1\underbrace{0\cdots0}\limits_{n-D-1}\oplus\\

&\ &\vdots\\

&\
&\underbrace{0\cdots0}\limits_{n_0}1\underbrace{0\cdots0}\limits_{n_1}
1\underbrace{0\cdots0}\limits_{n_2}1\underbrace{0\cdots0}\limits_{n-D-n_0}\\

&=&\underbrace{1\cdots1}\limits_{D+n_0}\|
\underbrace{0\cdots0}\limits_{n-D-n_0}.
\end{array}
\end{equation*}

Let
\begin{equation*}
\begin{array}{ll}
\mathbf{u}_0&=\bigoplus\limits_{0\leq i\leq
k-1}\rho^{i(D+n_0)}(\mathbf{u}_2)\\
&=\underbrace{1\cdots1}\limits_{D+n_0}\|\cdots\|\underbrace{1\cdots1}\limits_{D+n_0}
\|\underbrace{0\cdots0}\limits_{n-k(D+n_0)}\\
&=\underbrace{1\cdots1}\limits_{k(D+n_0)}\|\underbrace{0\cdots0}\limits_{n-k(D+n_0)},
\end{array}
\end{equation*}
where $1\leq k\leq \lfloor\frac{n-n_1-1}{D+n_0}\rfloor$. It is not
difficult to see that $v_2(h_f(\mathbf{u}_0))=k(n_0+1)$.

Let $k=\lfloor\frac{n-n_1-1}{D+n_0}\rfloor=q$. By Theorem
\ref{bentpanding}, we have
$v_2(h_f(\mathbf{u}_0))=q(n_0+1)>|\mathbf{u}_0|-n/2$, which
implies $q(D-n_0-1)< r+n_1+1$. It follows that $f$ is not bent if
$q(D-n_0-1)\geq r+n_1+1$.
\end{proof}

\begin{remark}
Note that the above nonexistence results could not be
obtained by Theorem \ref{dalai-stanica}. For example, the
nonexistence of homogeneous RotS bent functions with SANF
$x_1\cdots x_d+x_1\cdots x_{d-1}x_{d+1}$ could not be proven by
Theorem \ref{dalai-stanica}.
\end{remark}

\begin{remark}
We remark that the statement ``prove the nonexistence of homogeneous
RotS bent functions of degree $\geq 3$ on a single cycle(i.e. the
SANF is $\mathbf{x}^\mathbf{u}$ for some $\mathbf{u}$)" in
\cite{stanica2004} is incorrect. The proof is based on
the assumption that all RotS functions of a single cycle are
affinely equivalent to RotS functions with SANF $x_1 x_2\cdots
x_d$. In fact there are many RotS functions of a single cycle that
are not affinely equivalent to $x_1 x_2\cdots x_d$.
\end{remark}


In the following we will give a characterization of homogeneous
RotS bent function of degree $2$. First recall two basic results
about bent functions and circulant matrixes. A circulant matrix
over $\mathbb{F}_2$ is of the form
\begin{equation*}
\left(
\begin{array}{ll}
&\mathbf{a}_1\\ &\rho(\mathbf{a}_1)\\ &\vdots \\
&\rho^{n-1}(\mathbf{a}_1) \end{array} \right),
\end{equation*}
where $\mathbf{a}_1=(a_{11},\cdots,a_{1n})\in \mathbb{F}_2^n$. So
a circulant matrix can be represented by its first row
$\mathbf{a}_1$. Further a circulant matrix over
$\mathbb{F}_2$ can be represented by the polynomial $\sum\limits_{1\leq j\leq
n}a_{1j}x^{j-1}\in \mathbb{F}_2[x]$.

\begin{lemma}\label{quadratic}
Quadratic Boolean function $f(x_1,\cdots,x_n)=\sum\limits_{1\leq
i\leq j\leq n}a_{ij}x_ix_j+\sum\limits_{1\leq i\leq n}b_ix_i$ is
bent if and only if the matrix $(a_{ij})_{n\times n}$ is
nonsingular, where $a_{ij}=a_{ji}\in \mathbb{F}_2, a_{ii}=0, 1\leq
i,j\leq n$.
\end{lemma}

\begin{lemma}\label{circulant}
Circulant matrix $(a_{ij})_{n\times n}$ over $\mathbb{F}_2$ is
nonsingular if and only if the polynomials $\sum\limits_{1\leq
j\leq n}a_{1j}x^{j-1}$ and $x^n+1$ are relatively prime, i.e.
$GCD(\sum\limits_{1\leq j\leq n}a_{1j}x^{j-1},\ x^n+1)=1.$
\end{lemma}

It can be verified that $\sum\limits_{1\leq i\leq n}x_ix_{e-1+i}$
with $e>n/2$ is in fact equal to $\sum\limits_{1\leq i\leq
n}x_ix_{n-e+1+i}$ with $n-e+2\leq n/2+1$. So we can assume a
homogeneous RotS function of degree 2 has SANF form
$x_1x_{e_1}+\cdots+x_1x_{e_m}$, where $2\leq e_1<e_2\cdots
<e_m\leq n/2+1, m\leq n/2$.  Obviously, the associated matrix
$(a_{ij})_{n\times n}$ of $f$ is circulant,  with the first row
$(a_{11},\cdots,a_{1n})$ such that
$$a_{11}=0, a_{1e_i}=a_{1(n+2-e_i)}=1, 1\leq i\leq m, and\ a_{1j}=0\ if j\neq e_{i}\ or\ n+2-e_i.  $$
Thus the corresponding polynomial is $\sum\limits_{1\leq i\leq
m}(x^{e_i-1}+x^{n+1-e_i})$, where $x^{e_i-1}+x^{n+1-e_i}$ is
assumed to be $x^{n/2}$ if $e_i-1=n+1-e_i=n/2$. By Lemma
\ref{quadratic} and Lemma \ref{circulant}, we have
\begin{theorem}
Homogeneous RotS function $f$ of degree 2 described as above is
bent if and only if
$$GCD(\sum\limits_{1\leq i\leq m}(x^{e_i-1}+x^{n+1-e_i}),\ x^n+1)=1.$$
\end{theorem}

\begin{remark}
A necessary condition for $GCD(\sum\limits_{1\leq i\leq
m}(x^{e_i-1}+x^{n+1-e_i}),\ x^n+1)=1$ is $x^{n/2}$ should be
contained in $\sum\limits_{1\leq i\leq m}(x^{e_i-1}+x^{n+1-e_i})$,
i.e. the SANF of $f$ must contain $x_1x_{n/2+1}$. For example, all
homogeneous degree $2$ RotS bent functions in $8-$variables are
(expressed in SANF forms, see \cite{stanica}):
\begin{equation*}
\begin{array}{ll}
&x_1x_5;\
x_1x_2+x_1x_5;\ x_1x_3+x_1x_5;\ x_1x_4+x_1x_5;\ x_1x_2+x_1x_3+x_1x_5; \\
&x_1x_2+x_1x_4+x_1x_5;\ x_1x_3+x_1x_4+x_1x_5;\
x_1x_2+x_1x_3+x_1x_4+x_1x_5.
\end{array}
\end{equation*}
\end{remark}

\section{Conclusion} In this paper, we presented a different method
suitable for the existence problem of homogeneous rotation symmetric bent
functions, which leaded to some new results, and may be used to prove the nonexistence  of
most homogeneous rotation symmetric bent functions with degree
$>2$ once their SANFs are given. Since the conjecture is only partially proved, we expect a fully proof with the aid of our proposed method.

\end{document}